\documentclass[12pt,reqno]{article}

\usepackage[usenames]{color}
\usepackage{amssymb}
\usepackage{amsmath}
\usepackage{amsthm}
\usepackage{amsfonts}
\usepackage{amscd}
\usepackage{graphicx}

\usepackage[colorlinks=true,
linkcolor=webgreen,
filecolor=webbrown,
citecolor=webgreen]{hyperref}

\definecolor{webgreen}{rgb}{0,.5,0}
\definecolor{webbrown}{rgb}{.6,0,0}

\usepackage{color}
\usepackage{fullpage}
\usepackage{float}

\usepackage{graphics}
\usepackage{latexsym}
\usepackage{epsf}
\usepackage{breakurl}

\def\suchthat{\, : \,}
\DeclareMathOperator{\per}{per}

\begin{document}

\theoremstyle{plain}
\newtheorem{theorem}{Theorem}
\newtheorem{corollary}[theorem]{Corollary}
\newtheorem{lemma}[theorem]{Lemma}
\newtheorem{proposition}[theorem]{Proposition}

\newtheorem{definition}[theorem]{Definition}
\theoremstyle{definition}
\newtheorem{example}[theorem]{Example}
\newtheorem{conjecture}[theorem]{Conjecture}

\theoremstyle{remark}
\newtheorem{remark}[theorem]{Remark}

\title{Prefixes of the Fibonacci word}

\author{
Jeffrey Shallit\\
School of Computer Science\\
University of Waterloo \\
Waterloo, ON  N2L 3G1 \\
Canada\\
\href{mailto:shallit@uwaterloo.ca}{\tt shallit@uwaterloo.ca} }

\maketitle

\begin{abstract}
Mignosi, Restivo, and Salemi (1998) proved that for all $\epsilon > 0$
there exists an integer $N$ such that all 
prefixes of the Fibonacci word of length $\geq N$
contain a suffix of exponent $\alpha^2-\epsilon$,
where $\alpha = (1+\sqrt{5})/2$ is the golden ratio.
In this note we show how to prove an explicit version of this
theorem with tools from
automata theory and logic.  Along the way we gain a better
understanding of the repetitive structure of the Fibonacci word.
\end{abstract}

\section{Introduction}

The Fibonacci word ${\bf f}[0..\infty) = 01001010 \cdots$ is a
well-studied 
infinite binary word with many interesting properties \cite{Berstel:1986b}.
It can be defined in many ways; for example, as the limit of the
sequence of finite words $(X_i)_{i \geq 1}$, with $X_1 = 1$, $X_2 = 0$,
and $X_n = X_{n-1} X_{n-2}$ for $n \geq 3$.  

Let $x = x[0..n-1]$ be a nonempty finite word.
We say $x$ has {\it period\/} $p$ if
$x[i]=x[i+p]$ for $0 \leq i < n-p$.   The least positive period of
a word $x$ is called {\it the\/} period, and is denoted $\per(x)$.   The
{\it exponent\/} of a word, denoted $\exp(x)$, is defined to be $|x|/\per(x)$.  Thus
{\tt entente} has exponent $7/3$.   We say a word $x$ is 
an {\it $\alpha$-power}, for $\alpha \geq 1$ a real number, if
$|x| = \lceil \alpha \per(x) \rceil$.   Thus, for example,
{\tt underfund} is a $\sqrt{2}$-power.

Let $\alpha = (1+\sqrt{5})/2$ be the golden ratio.
In 1997 this author conjectured, and Mignosi, Restivo, and Salemi
later proved \cite{Mignosi&Restivo&Salemi:1998}, that
an infinite word $\bf x$ has the property
that every sufficiently long prefix
has a suffix that is an $\alpha^2$-power if and only if $\bf x$ is ultimately
periodic.  This is an example of the ``local periodicity implies
global periodicity'' phenomenon.

Furthermore, they proved that the constant $\alpha^2$ is best possible.
Define $e(n)$ to be the largest exponent of a suffix of ${\bf f}[0..n-1]$ and
for a real number $\gamma$ 
define $M_\gamma = \{ n \suchthat e(n)\geq \gamma \} $.
Mignosi et al.~proved that for all $\epsilon > 0$,
the set $M_{\alpha^2- \epsilon}$ contains all but finitely many natural
numbers.

This latter result was proved as their Proposition 2, using a 
somewhat complicated induction.  In this note we
show how to prove an explicit version of the result using {\tt Walnut}, 
a free theorem-prover for automatic sequences \cite{Mousavi:2016,Shallit:2022}.
We do not claim our proof to be simpler than the previous one,
but it does give a more explicit
version of the theorem, and is more straightforward in some sense.
As a bonus, we gain a more detailed understanding of the repetitive
properties of the prefixes of $\bf f$.

\section{Prefixes}

Define, as usual, the Fibonacci numbers
by the recurrence $F_n = F_{n-1} + F_{n-2}$ with initial values
$F_0 = 0$, $F_1 = 1$, and the Lucas numbers
$L_n = L_{n-1} + L_{n-2}$ with initial values
$L_0 = 2$, $L_1 = 1$.
By the Binet formulas we have
$F_n = (\alpha^n - \beta^n)/\sqrt{5}$, where
$\beta = (1-\sqrt{5})/2$, and
$L_n = \alpha^n + \beta^n$, which also serve to define
$F_n$ and $L_n$ for negative indices. Note that
$F_{-n} = (-1)^{n+1} F_n$ and
$L_{-n} = (-1)^{n+1} L_n$ 
for all integers $n$.

By a result of Currie and Saari \cite[Cor.~4]{Currie&Saari:2008},
we know that all (least) periods of the
Fibonacci word are of length $F_n$ for $n \geq 2$.  
(This can also be proved with {\tt Walnut}; see 
\cite[Thm.~3.15]{Mousavi&Schaeffer&Shallit:2016}.)
Thus, if we are to look
for large repetitions, we can restrict our attention to 
periods of length a Fibonacci number.

We start by determining $G := M_{\alpha^2}$: that is,
those ``good'' $n$ for which ${\bf f}[0..n-1]$
has a suffix of exponent $> \alpha^2$.  
With {\tt Walnut} we can characterize
the set $$G =
\{13, 14, 22, 23, 24, 26, 27, 34, 35, 36, 37, 38, 39, 40, 43, \ldots \}.$$
To do so, we need
to phrase the assertions as a first-order logical formula in a 
language {\tt Walnut} can understand.
We represent numbers in the so-called
Fibonacci (or Zeckendorf) numeration systems;
see \cite{Lekkerkerker:1952} or \cite{Zeckendorf:1972} for more details.
Here is the {\tt Walnut} code:
\begin{verbatim}
reg isfib msd_fib "0*10*":
reg evenfib msd_fib "0*1(00)*":
reg oddfib msd_fib "0*10(00)*":
reg adjfib msd_fib msd_fib "([0,0]*[1,1])|[0,0]*[1,0][0,1][0,0]*":
def ffactoreq "?msd_fib At (t<n) => F[i+t]=F[j+t]":
def suff "?msd_fib n>=x & x>=y & y>=1 & $ffactoreq(n-x,(n+y)-x,x-y)":
reg shift {0,1} {0,1} "([0,0]|[0,1][1,1]*[1,0])*":
def phi2n "?msd_fib (s=0 & n=0) | Ex,y $shift(n-1,x) &
   $shift(x,y) & s=y+2":
def good "?msd_fib Ex,y,z $suff(n,x,y) & $phi2n(y,z) & x>z":
\end{verbatim}
The explanation is as follows:
\begin{itemize}
\item {\tt isfib} asserts that its argument is a positive Fibonacci
number.

\item {\tt evenfib} asserts that its argument is a positive
even-indexed Fibonacci number.

\item {\tt oddfib} asserts that its argument is a positive
odd-indexed Fibonacci number.

\item {\tt adjfib} asserts that its two arguments are
$F_{k+1}$ and $F_k$ for some integer $k \geq 1$.

\item {\tt factoreq} asserts that ${\bf f}[i..i+n-1] = {\bf f}[j..j+n-1]$.

\item {\tt phi2n} is code from \cite[p.~278]{Shallit:2022} that
asserts that $s = \lfloor \alpha^2 n \rfloor$.

\item {\tt suff} asserts that ${\bf f}[0..n-1]$ has a suffix
of length $x$ with period $y \geq 1$.

\item {\tt good} asserts that ${\bf f}[0..n-1]$ has a suffix
with exponent $> \alpha^2$.
\end{itemize}

The last command produces a DFA (deterministic finite automaton) of 12
states accepting the Fibonacci representation of those $n$
belonging to $G$, in Figure~\ref{good}.
\begin{figure}[H]
\begin{center}
\includegraphics[width=6.5in]{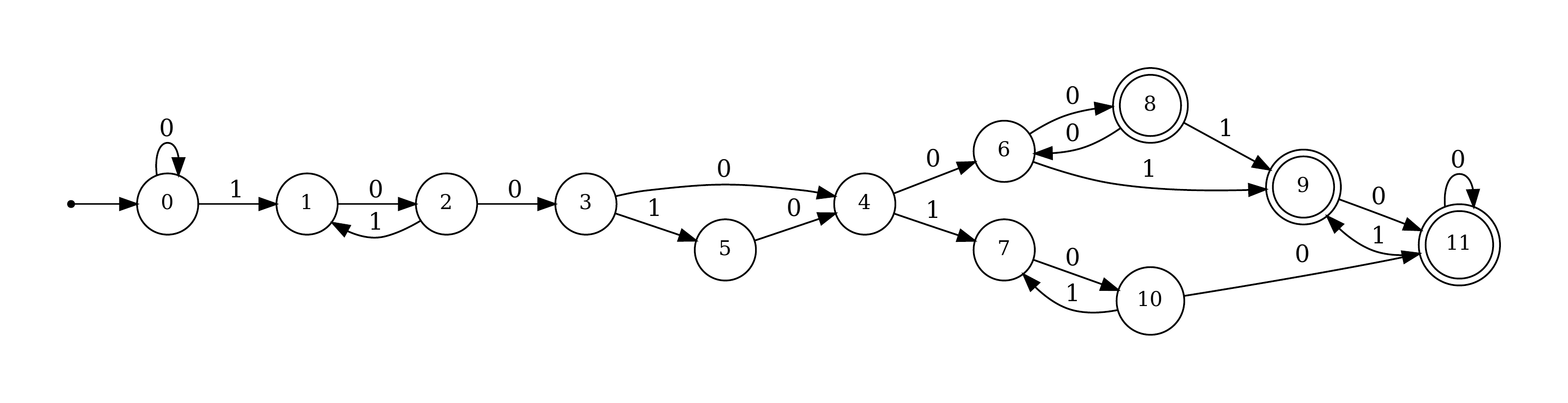}
\end{center}
\caption{Fibonacci automaton recognizing the set of good $n$.}
\label{good}
\end{figure}

We claim that the only $n\geq 2$ that are not in $G$ are either 
\begin{itemize}
\item[(1)] of the form $n = F_i - F_j - 1$ for $i \geq 5$ and $3 \leq j \leq i-2$; or
\item[(2)] of the form $n = F_i - F_{2j+1}$ for $i \geq 5$
and $1 \leq j \leq (i-3)/2$.
\end{itemize}
Let $B_1$ be the set of $n$ satisfying (1) and $B_2$ be the set of
$n$ satisfying (2).  Here are the first few terms of each set:
\begin{align*}
B_1 &= \{ 2, 4, 5, 7, 9, 10, 12, 15, 17, 18, 20, 25, 28, 30, 31, 33, \ldots \} \\
B_2 &= \{ 3, 6, 8, 11, 16, 19, 21, 29, 32, 42, 50, 53, 55, 76, 84, 87, \ldots \} .
\end{align*}

We can verify the claim above as follows.  For (1) we assert
that $x = F_i$ and $y = F_j$ for some $i \geq j+2$ and $j \geq 3$.
This is the same as asserting that $(x,y)$ are adjacent
Fibonacci numbers with $x > 2y$ and $y \geq 2$.
For (2) we let $x = F_i$ and $y = F_{2j+1}$ for some $i \geq 2j+3$ and $j \geq 1$.
This is the same as asserting that $(x,y)$ are adjacent
Fibonacci numbers, $y$ is an odd-indexed Fibonacci number, $x>2y$, and $y \geq 2$.
\begin{verbatim}
def b1 "?msd_fib $isfib(x) & $isfib(y) & x>2*y & y>=2 & n+1+y=x":
def b2 "?msd_fib $isfib(x) & $oddfib(y) & x>2*y & n>=y & y>=2 & n+y=x":
\end{verbatim}
We now verify that for all $n \geq 2$ a number is in $G$ if and only
if it is in neither $B_1$ nor $B_2$.
\begin{verbatim}
eval test "?msd_fib An (n>=2) => ( (~$good(n)) <=> 
   (Ex,y $b1(n,x,y)|$b2(n,x,y)) )":
\end{verbatim}
And {\tt Walnut} returns {\tt TRUE}.

For the $n$ satisfying either condition (1) or (2), we need to locate
a suffix of ${\bf f}[0..n-1]$ of large exponent.
First we show
\begin{lemma}
Suppose $n = F_i - F_j - 1$ for $i \geq 5$ and $3 \leq j \leq i-2$.   Then
${\bf f}[0..n-1]$ 
\begin{itemize}
\item[(i)] has period $F_{i-2}$, or
\item[(ii)] has a suffix of length $F_j - 1$ and period $F_{j-2}$.
\end{itemize}
\label{lem1}
\end{lemma}

\begin{proof}
The idea is to use $x$ to represent $F_{i-1}$, $y$ to represent $F_{i-2}$,
$w$ to represent $F_{j-1}$ and $z$ to represent $F_{j-2}$.
We use the command
\begin{verbatim}
eval check1 "?msd_fib An,x,y,z,w (n>=2 & $b1(n,x+y,w+z) & $adjfib(x,y) &
  $adjfib(w,z) & z>=1 & x>w & n+w+z+1=x+y) 
  => ($suff(n,w+z-1,z) & $suff(n,n,y))":
\end{verbatim}
and {\tt Walnut} returns {\tt TRUE}.
\end{proof}

\begin{lemma}
Suppose $n = F_i - F_{2j+1}$ for $i \geq 5$ and $1 \leq j \leq (i-3)/2$.
Then ${\bf f}[0..n-1]$ has period $F_{i-2}$.  If further $j \geq 2$ then
${\bf f}[0..n-1]$ has a suffix of length $F_{2j+1}$ and
period $F_{2j-1}$.
\label{lem2}
\end{lemma}
\begin{proof}
The idea is to use $x$ to represent $F_{i-1}$, $y$ to represent $F_{i-2}$,
$w$ to represent $F_{2j}$ and $z$ to represent $F_{2j-1}$.
We use the commands
\begin{verbatim}
eval check2a "?msd_fib An,x,y,z,w (n>=3 & $b2(n,x+y,w+z) & $adjfib(x,y) 
   & $adjfib(w,z) & n+w+z=x+y) => $suff(n,n,y)":
eval check2b "?msd_fib An,x,y,z,w (n>=3 & $b2(n,x+y,w+z) & $adjfib(x,y)
   & $adjfib(w,z) & n+w+z=x+y & w>=2) => $suff(n,w+z,z)":
\end{verbatim}
and {\tt Walnut} returns {\tt TRUE} for both.
\end{proof}

Our main result will now follow after some detailed estimates in
the next section.

\section{The main result}

We will need a classic identity about Fibonacci numbers \cite[p.~68]{Basin&Hoggatt:1963}, which
can be easily proved with the Binet formula, namely,
\begin{equation}
F_a F_{b+2} - F_{a+2} F_b  = (-1)^b F_{a-b} 
\label{useful}
\end{equation}
for all integers $a, b$.

Next, we need some estimates.
\begin{lemma}
For $i \geq 1$ we have
\begin{alignat}{2}
\alpha^2 + { {(-1)^i}  \over {F_{2i}}} &< \hphantom{ck} {{F_{i+2}} \over {F_i}} 
&&<\alpha^2 + { {(-1)^i}  \over { F_{2i} - (-1)^i}}  \label{eq1} \\
\alpha^2 + {1 \over {F_{2i-1}+2}} &< {{F_{2i+1} + 1}\over{F_{2i-1}}} &&<
\alpha^2 + {1 \over {F_{2i-1}}} \label{eq2} .
\end{alignat}
\end{lemma}

\begin{proof}
Directly follows from the Binet formula.
\end{proof}

Finally, we need some technical estimates.
\begin{lemma}
\leavevmode
\begin{itemize}
\item[(i)] $(F_{k+1} + 1)^2 \leq 3 F_{2k-1}$ for $k \geq 4$;
\item[(ii)] $F_{4k-2}^2 \geq 100 F_{2k+1}$ for $k \geq 3$.
\item[(iii)] $F_{2k+2} \leq 6 F_k^2$ for $k \geq 5$.
\item[(iv)] $F_{2k}^2 \geq 8 F_{2k+2}$ for $k \geq 4$.
\item[(v)] $F_{12k-4}^2 \geq 10000 F_{6k+3}$ for $k \geq 2$.
\item[(vi)] $F_{2k+1}^2 F_{6k+3} \leq 6 F_{6k-2}^2$ for 
$k \geq 2$.
\item[(vii)] $F_{4k+2}^2 \geq 4F_{6k+5}$ for $k\geq 3$.
\end{itemize}
\label{bnds}
\end{lemma}

\begin{proof}
We prove only (i) and (iii); the others are simpler and are left to the reader.
From the Binet form, we can verify that
$$3 F_{2k-1} - (F_{k+1} + 1)^2 = {1\over 5} (8 F_{2k-3} + 4 F_{2k-2} - 10 F_{k+1} -5 -2(-1)^k) \geq 0$$
for $k \geq 4$.  This proves (i).

Similarly, we can verify that
$$ 6F_k^2 - F_{2k+2} = {1\over 5} (3 F_{2k-1} - 4F_{2k-2} -12 (-1)^k) \geq 0$$
for $k \geq 5$.  This proves (iii).
\end{proof}

\begin{theorem}
For all $n \geq 1$ the prefix ${\bf f}[0..n-1]$ 
has a suffix of exponent $\geq \alpha^2 - 3 n^{-1/2}$.
\label{mainthm}
\end{theorem}

\begin{proof}
We can check the result for $1 \leq n \leq 21$. In what follows, then,
we assume $n > 21$.
There are three cases to consider:

\medskip

\noindent {\it Case 1:}  $n\in G$.  In this case, we have 
already seen above that there is a suffix
of ${\bf f}[0..n-1]$ of exponent $> \alpha^2$.

\medskip

\noindent{\it Case 2:}  $n \in B_1$, that is,
$n = n_{i,j} = F_i - F_j - 1$ for $i \geq 5$ and $3 \leq j \leq i-2$.
Fix $i$; using Lemma~\ref{lem1} 
we will find a lower bound for $e(n_{i,j})$ applicable
to all $j$, $3 \leq j \leq i-2$.  Note that $n_{i,j} \leq F_i$.
Define
\begin{align*}
f(i,j) &= (F_j - 1)/F_{j-2} \\
g(i,j) &= (F_i - F_j - 1)/F_{i-2} \\
\rho(i,j) &= (F_i - F_j - 1) F_{j-2} - (F_j -1) F_{i-2}.
\end{align*}
Clearly $g(i,j) \geq f(i,j)$ iff $\rho(i,j) \geq 0$.

We now claim that $f$ is increasing with increasing $j$.
To see this, note that
\begin{align*}
f(i,j+1) - f(i,j) &= {{F_{j+1} - 1} \over {F_{j-1}}} - {{F_j - 1} \over {F_{j-2}} }\\
&= {{F_{j+1} F_{j-2} -F_{j-1}F_j + F_{j-1} - F_{j-2}} \over {F_{j-1} F_{j-2}}} \\
&= {{ (-1)^{j+1} + F_{j-3} \over {F_{j-1} F_{j-2}}}} \geq 0
\end{align*}
for $j \geq 3$.  Here we have used
Eq.~\eqref{useful} with $a = j-2$ and $b = j-1$.

Similarly, $g$ is decreasing with increasing $j$.
To see this, note that
$$ g(i,j+1)-g(i,j) = {{F_j - F_{j+1}} \over {F_{i-2}}} \leq 0$$
for $i \geq 3$.  Hence, for each fixed $i \geq 3$,
there is a ``crossover point'' where the graphs of
$f(i,j)$ and $g(i,j)$ cross.   This crossover point determines
where $\min_j \max( f(i,j), g(i,j) )$ occurs.  See
Figure~\ref{fig2} for an example for $i = 20$.
\begin{figure}[H]
\begin{center}
\includegraphics[width=5.5in]{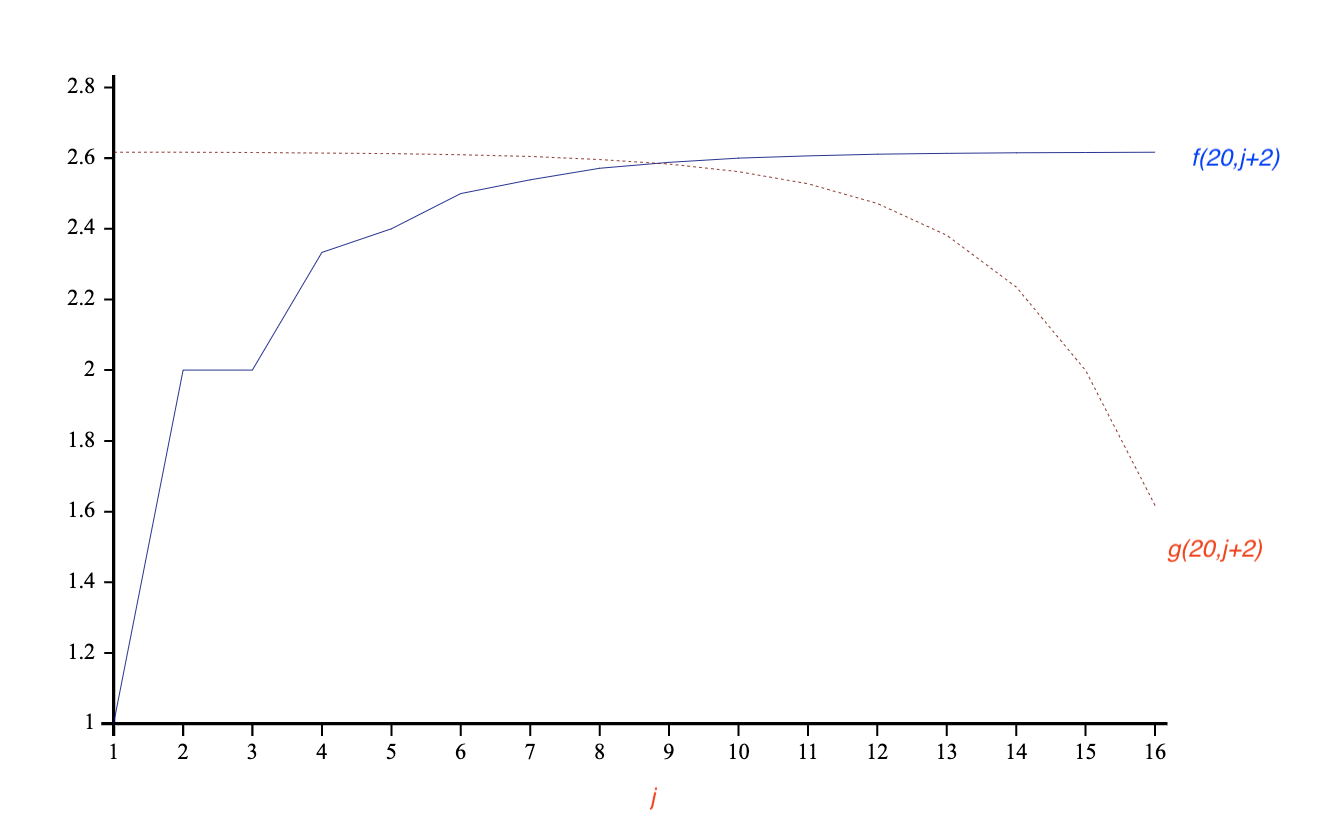}
\end{center}
\vskip-.3in
\caption{The crossover point for $f(20,j+2)$ versus $g(20,j+2)$.}
\label{fig2}
\end{figure}

Observe that $\rho(i,j)$ is increasing with increasing $i$:
\begin{align*}
\rho(i+1,j)-\rho(i,j) &= (F_{i+1} - F_j - 1)F_{j-2} - (F_j - 1)F_{i-1}
- (F_i - F_j - 1)F_{j-2} + (F_j - 1) F_{i-2} \\
&= F_{i-1} F_{j-2} - (F_j - 1) F_{i-3} \\
&= F_{i-3} - (-1)^j F_{i-j-1} \geq 0.
\end{align*}
Here we used Eq.~\eqref{useful} with $a = i-3$ and $b = j-2$.

Using the Binet form one can verify that
\begin{align*}
\rho(2k+1,k+1) &= {1\over 5}(L_{2k-2} - 5F_{k-1} - 5 F_{-k} - 3(-1)^k) \\
\rho(2k+2,k+2) &= {1\over 5} (-L_{2k-2} - 5F_k + 5 F_{-k} + 3(-1)^k) .
\end{align*}
It follows that $\rho(2k+2,k+1) \geq \rho(2k+1,k+1) > 0$ for $k \geq 4$ and
and $\rho(2k+1,k+2) \leq \rho(2k+2,k+2) < 0$ for $k \geq 1$.
Set $j' = \lceil i/2 \rceil$.  
Hence $g(i,j') - f(i,j') \geq 0$ for $i \geq 6$ and
$g(i,j'+1) - f(i,j'+1) < 0$ for $i \geq 0$.   
It follows that if $n \in B_1$, and $n = F_i - F_j - 1$ for some $i \geq 8$,
then $e(n) \geq \min( g(i, j'), f(i, j' + 1))$.

There are now two subcases to consider:  $i$ is odd and $i$ is even.

Suppose $i$ is odd.
Then $i = 2k+1$ and $j' = k+1$ for some $k$.
Then using the Binet form one can verify that
$$ F_k F_{2k-1} (f(2k+1,k+2)-g(2k+1,k+1)) = 
 {1\over 10} ( -2(-1)^k + 2L_{2k-3} + 10 F_k + 5 F_{-k} + 5 L_{-k}) > 0$$
 for $k \geq 0$.
Hence $e(n) \geq g(2k+1,k+1)$.

From Lemma~\ref{bnds} (i) we know that $(F_{k+1} + 1)^2/F_{2k-1} \leq 3$ for
$k \geq 4$.  From Eq.~\eqref{eq1} we know that
$F_{2k+1}/F_{2k-1} \leq \alpha^2$.
Multiplying these together and dividing by $F_{2k+1}$, we get
$$(F_{k+1} +1)^2  /F_{2k-1}^2 \leq 3 \alpha^2 F_{2k+1}^{-1}
	\leq 3 \alpha^2 n^{-1}$$
for $k \geq 4$.
Hence 
$$(F_{k+1} +1)/F_{2k-1} \leq \sqrt{3} \alpha n^{-1/2}$$
for $k \geq 3$.
Then 
\begin{align*}
e(n) \geq g(2k+1,k+1) &= {{F_{2k+1}}  \over {F_{2k-1}}} -
	{{F_{k+1} + 1} \over {F_{2k-1}}}     \\
&\geq (\alpha^2 - {1 \over {F_{4k-2}}}) - \sqrt{3} \alpha n^{-1/2} \\
&\geq \alpha^2 - {1\over 10} n^{-1/2} - \sqrt{3} \alpha n^{-1/2}  
	\quad \text{(by Lemma~\ref{bnds} (ii))} \\
& \geq \alpha^2 - 3 n^{-1/2}
\end{align*}
for $k \geq 2$.

If $i$ is even then $i = 2k+2$ and $j' = k+1$.
Then using the Binet
form one can verify that
$$ F_k F_{2k} (g(2k+2,k+1) - f(2k+2,k+2)) =
 {1 \over 10} ( 2(-1)^k + 2 L_{2k-1} - 10F_k - 5 F_{1-k} + 5L_{1-k}) \geq 0$$
for $k \geq 3$.
Hence $e(n) \geq f(2k+2,k+2)$.

Now from Lemma~\ref{bnds} (iv) we get $-1/F_{2k} \geq -8^{-1/2} F_{2k+2}^{-1/2} \geq -8^{-1/2} n^{-1/2}$.
Similarly from Lemma~\ref{bnds} (iii) we get
$-1/F_k \geq -\sqrt{6} F_{2k+2}^{-1/2} \geq - \sqrt{6}  n^{-1/2}$.  Hence
\begin{align*}
e(n) \geq f(2k+2,k+2) &= {{F_{k+2} - 1} \over {F_k}} \\
&\geq (\alpha^2 - {1 \over F_{2k}}) - {1 \over {F_k}} \\
&\geq \alpha^2 - 8^{-1/2} n^{-1/2} - \sqrt{6} n^{-1/2} \\
& \geq \alpha^2 - 3n^{-1/2}.
\end{align*}

\medskip

\noindent{\it Case 3:}  $n \in B_2$, that is, $n = n_{i,j} = F_i - F_{2j+1}$
for $i \geq 5$ and $1 \leq j \leq (i-3)/2$.  (For these $n$ it turns
out that $e(n) \geq \alpha^2 - O(n^{-2/3})$, but for our theorem
it is not necessary to prove this stronger result.)

Using Lemma~\ref{lem2}, we find lower bound for $e(n_{i,j})$ applicable to all
$j$, $1 \leq j \leq (i-3)/2$.
Define
\begin{align*}
r(i,j) &= F_{2j+1}/F_{2j-1} \\
s(i,j)  &= (F_i - F_{2j+1})/F_{i-2} \\
\psi(i,j) &= (F_i - F_{2j+1})F_{2j-1} - F_{i-2} F_{2j+1} .
\end{align*}
Again we need to determine the ``crossover point'', illustrated in
Figure~\ref{fig3}.
\begin{figure}[htb]
\begin{center}
\includegraphics[width=6.1in]{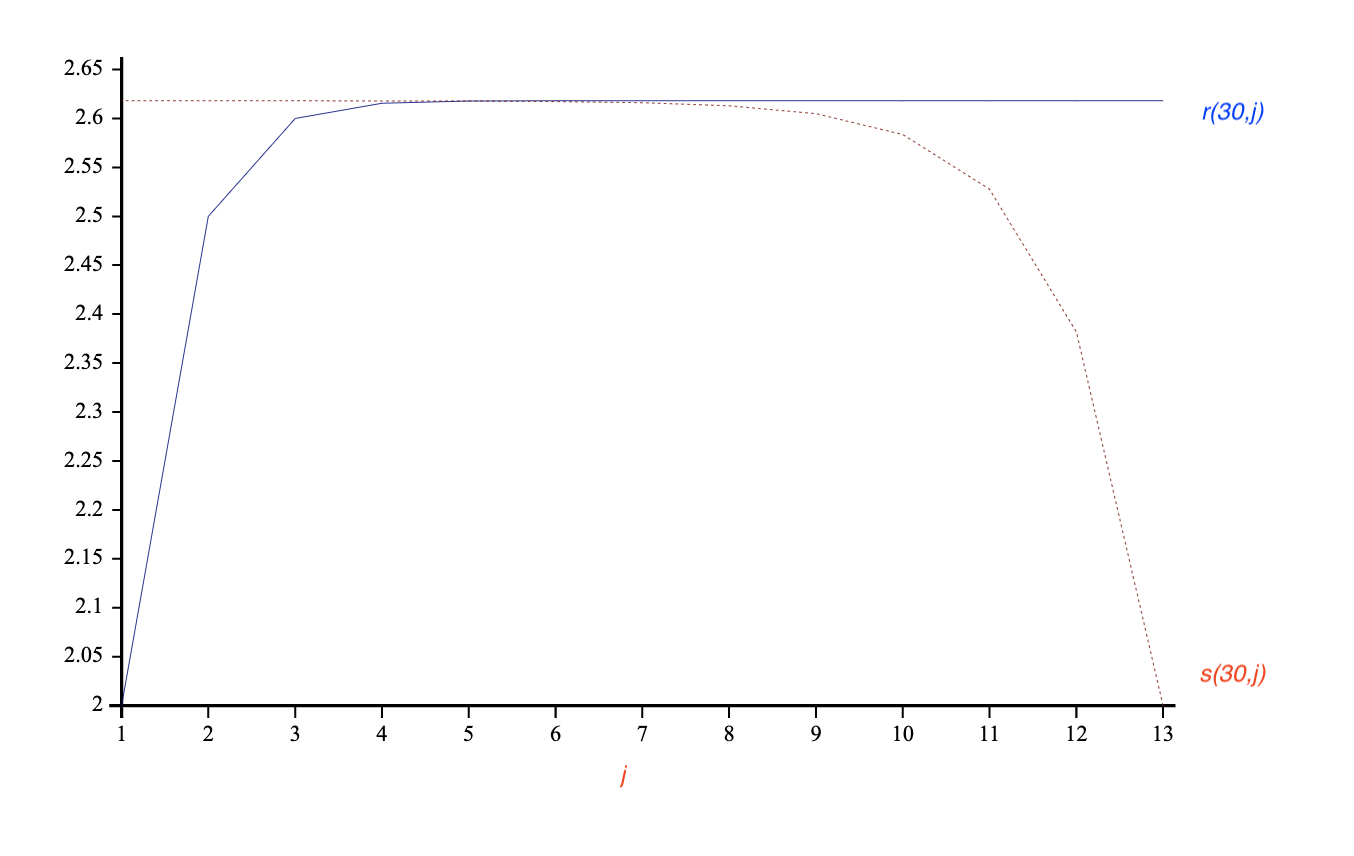}
\end{center}
\vskip -.3in
\caption{The crossover point for $r(30,j)$ versus $s(30,j)$.}
\label{fig3}
\end{figure}

Clearly $r$ is increasing with increasing $j$.  To see
this, note that
\begin{align*}
r(i,j+1) - r(i,j) &= {{F_{2j+3}} \over {F_{2j+1}}} -
	{{F_{2j+1}} \over {F_{2j-1}}}  \\
&= {{F_{2j+3} F_{2j-1} - F_{2j+1}^2} \over {F_{2j+1} F_{2j-1}}} \\
&= {1 \over {{F_{2j+1} F_{2j-1}}}} > 0.
\end{align*}
Here we used Eq.~\eqref{useful} with $a = 2j-1$ and $b = 2j+1$.

Next, we show that $s$ is decreasing with increasing $j$:
\begin{align*}
s(i,j+1)-s(i,j) &= {{F_i - F_{2j+3}} \over {F_{i-2}}} 
- {{F_i - F_{2j+1}} \over {F_{i-2}}} \\
&= {{F_{2j+1} - F_{2j+3}} \over {F_{i-2}}} < 0,
\end{align*}
where we used Eq.~\eqref{useful} with $a = 2j-1$ and $ b=i-3$.

Also $\psi(i,j)$ is increasing with increasing $i$,
provided $i \geq 2j+2$.  We have
\begin{align}
\psi(i+1,j) - \psi(i,j) &=
	(F_{i+1} - F_{2j+1})F_{2j-1} - F_{i-1} F_{2j+1} 
	- (F_i - F_{2j+1}) F_{2j-1} - F_{i-2} F_{2j+1} \nonumber\\
&= F_{i-1} F_{2j-1} - F_{i-3} F_{2j+1} \label{psi_inc} \\
&= F_{i-2j-2} \geq 0 . \nonumber
\end{align}

Using the Binet form we can easily check that
\begin{align*}
\psi(6k,k) &= (L_{4k-2} - 3)/5 \\
\psi(6k+5,k+1) &= (-L_{4k} -3)/5 .
\end{align*}

This, combined with \eqref{psi_inc} above, gives
$$
\psi(6k+5,k) \geq \psi(6k+4,k) \geq \psi(6k+3,k)
\geq \psi(6k+2,k) \geq \psi(6k+1,k) \geq \psi(6k,k) \geq 0 $$
and 
\begin{multline*}
\psi(6k,k+1) \leq \psi(6k+1,k+1) \leq \psi(6k+2,k+1) 
\leq \psi(6k+3,k+1) \leq \psi(6k+4,k+1)  \\
\leq \psi(6k+5,k+1) \leq 0.
\end{multline*}

Now choose $j' = \lfloor i/6 \rfloor$.
Then $r(i,j') \leq s(i,j')$ and
$r(i,j'+1) \geq s(i, j'+1)$ for $i \geq 1$.
So for $i \geq 1$ we have $e(n) \geq \min(s(i,j'), r(i,j'+1))$.

There are now two cases to consider:     
\begin{itemize}
\item[(a)] $i = 6k+a$, $a \in \{0,1,2,3 \}$, and $j = k$, and
\item[(b)] $i = 6k+a$, $ a \in \{4,5 \}$, and $j = k$.
\end{itemize}

In case (a) we have, for $i = 6k$, that
\begin{align*}
F_{2k+1} F_{6k-2} (r(6k,k+1)-s(6k,k)) 
&= F_{6k-2} F_{2k+3} - F_{2k+1} F_{6k} + F_{2k+1}^2 \\
&= {1 \over {20}} (7 L_{4k-3} + 15 F_{4k} + 8) \geq 0 . 
\end{align*}

For $i = 6k+1$:
\begin{align*}
F_{2k+1} F_{6k-1} (r(6k+1,k+1)-s(6k+1,k)) 
&= F_{6k-1} F_{2k+3} - F_{2k+1} F_{6k+1} +F_{2k+1}^2 \\
&= {1 \over 5} (-2L_{4k-3} + 5 F_{4k} + 2) \geq 0. 
\end{align*}

For $i = 6k+2$:
\begin{align*}
F_{2k+1} F_{6k} (r(6k+2,k+1) - s(6k+2,k))
&= F_{6k} F_{2k+3} - F_{2k+1} F_{6k+2} + F_{2k+1}^2 \\
&= {1 \over {10}} (L_{4k-3} + 5 F_{4k} + 4) \geq 0. 
\end{align*}

For $i = 6k+3$:
\begin{align*}
F_{2k+1} F_{6k+1} (r(6k+3,k+1)- s(6k+3,k))
&= F_{6k+1} F_{2k+3} - F_{2k+1} F_{6k+3} + F_{2k+1}^2 \\
&= {1 \over {20}} (-3L_{4k-3} + 5F_{4k} + 8) \geq 0.
\end{align*}

Hence
\begin{align*}
e(n) \geq s(6k+a,k) &= {{F_{6k+a} - F_{2k+1}} \over {F_{6k+a-2}}} \\
&\geq (\alpha^2 - {1 \over {F_{2k+2a-4}}} ) -
{{F_{2k+1}} \over {F_{6k+a-2}}} \\
&\geq \alpha^2 - (1/100) F_{6k+3}^{-1/2} - \sqrt{6} F_{6k+a}^{-1/2} 
	\quad \text{(by Lemma~\ref{bnds} (vi))} \\
&\geq \alpha^2 - (1/100+\sqrt{6}) n^{-1/2} \\
&\geq \alpha^2 - 3n^{-1/2}. 
\end{align*}

In case (b) we have $a \in \{4,5 \}$.  Then
\begin{align*}
e(n) \geq r(6k+a,k+1) &= {{F_{2k+3}} \over {F_{2k+1}}}  \\
&\geq \alpha^2 - {1 \over {F_{4k+2}}} \\
&\geq \alpha^2 - {1\over 2} F_{6k+5}^{-1/2} \\
&\geq \alpha^2 - {1 \over 2} F_{6k+a}^{-1/2} \\
& \geq \alpha^2 - 2 n^{-1/2} \quad \text{(by Lemma~\ref{bnds} (vii))} 
\end{align*}
for $k \geq 3$.

Thus in all cases $e(n) \geq \alpha^2 - 3n^{-1/2}$.  This completes the proof.
\end{proof}

\section{Final remarks}

For real numbers $\gamma$,
the set $M_{\gamma}$ can be quite complicated.  For example,
Rampersad \cite{Rampersad:2023} recently studied the set $M_3$.
In {\tt Walnut} one can create a Fibonacci automaton accepting
the set $M_{p/q}$, for natural numbers $p, q$, with a command of
the form
\begin{verbatim}
def emmpq "?msd_fib Ex,y $suff(n,x,y) & p*y<=q*x":
\end{verbatim}
where {\tt p} is replaced by the specific number $p$
and {\tt q} is replaced by $q$.

Given a rational number $p/q < \alpha^2$, {\tt Walnut}
can also determine the largest $n$ for
which $e(n) < p/q$, as follows:
\begin{verbatim}
def has_suff "?msd_fib Ex,y $suff(n,x,y) & p*y<q*x":
def largest_index "?msd_fib (~$has_suff(n)) & Am (m>n) => $has_suff(m)":
test largest_index 1:
\end{verbatim}
This code prints the Fibonacci representation of $n$. Again, one 
needs to replace {\tt p} with the specific number $p$ and
likewise with ${\tt q}$.

With more work, it should be possible to prove that, for $k \geq 6$, the largest
$n$ for which $e(n) < (F_{k+1}-1)/F_{k-1} $ is $n = F_{2k-1} - F_k - 1$.

\end{document}